\documentclass{article}
\usepackage{arxiv}
\usepackage[utf8]{inputenc}
\usepackage[T1]{fontenc}
\usepackage{url}
\usepackage{mathptmx}
\usepackage{amsmath}
\usepackage{amsthm}
\usepackage{amssymb}
\usepackage[hidelinks]{hyperref}

\newtheoremstyle{thm-break}{10pt}{3pt}{\normalfont}{}{\bfseries}{}{\newline}{\thmname{#1}\thmnumber{ #2}\thmnote{ (#3)}\par\vspace{0.5\baselineskip}}
\newtheoremstyle{remark-example-space}{10pt}{3pt}{\normalfont}{}{\bfseries}{.}{0.5em}{}
\theoremstyle{thm-break}
\newtheorem{theorem}{Theorem}[section]
\theoremstyle{remark-example-space}
\newtheorem{remark}[theorem]{Remark}
\theoremstyle{thm-break}
\newtheorem{definition}[theorem]{Definition}
\theoremstyle{thm-break}
\newtheorem{proposition}[theorem]{Proposition}
\theoremstyle{thm-break}
\newtheorem{lemma}[theorem]{Lemma}
\theoremstyle{thm-break}
\newtheorem{corollary}[theorem]{Corollary}
\theoremstyle{remark-example-space}
\newtheorem{example}[theorem]{Example}

\newcounter{casei}
\newlength{\caselabelwidth}
\newenvironment{caselist}{%
  \begin{list}{(\roman{casei})}{%
    \usecounter{casei}%
    \setlength{\labelwidth}{\caselabelwidth}%
    \setlength{\labelsep}{0.4em}%
    \setlength{\leftmargin}{\dimexpr\caselabelwidth+\labelsep\relax}%
    \setlength{\itemsep}{0.35\baselineskip}%
    \setlength{\topsep}{0.35\baselineskip}%
    \setlength{\parsep}{0pt}%
    \setlength{\partopsep}{0pt}%
  }%
}{\end{list}}
\providecommand{\keywords}[1]{}
\renewcommand{\keywords}[1]{%
\begin{center}
\small\textbf{Keywords:} #1
\end{center}
}
\renewenvironment{proof}[1][Proof]{\par\noindent\textbf{#1. }\normalfont}{\hfill$\square$\par}
\newcommand{\doi}[1]{\href{https://doi.org/#1}{\nolinkurl{https://doi.org/#1}}}
\title{Spectral Sign Implication for Quantum Logic}
\author{%
\href{https://orcid.org/0000-0001-8878-9824}{Kenji Tokuo}\\
{\normalfont Department of Information Engineering, Oita College}\\
{\normalfont National Institute of Technology}\\
{\normalfont Oita 870-0152, Japan}\\
{\normalfont \texttt{tokuo@oita-ct.ac.jp}}%
}
\date{}
\renewcommand{\undertitle}{A preprint}

\begin{document}
\maketitle

\begin{abstract}
We study the spectral sign implication for quantum logic, defined by the nonnegative spectral projection of the operator obtained by subtracting the antecedent projection from the consequent projection. The construction agrees with classical material implication on commuting projections and compares arbitrary pairs of projections through the spectral structure of their difference. It satisfies Hardegree's four minimal implicative conditions, his law of contraposition, and a falsity condition. It differs from the standard polynomial implications in that its value need not belong to the ortholattice generated by its arguments. Within a uniform class of Borel constructions for pairs of projections, the operations satisfying entailment, contraposition, and the falsity condition correspond exactly to measurable choices of spectral branch. In the continuous subclass, these three conditions determine the spectral sign implication uniquely. In finite dimensions, the operation is also the largest among the acceptance projections of optimal projective tests for Helstrom discrimination between subspace states with priors proportional to rank. These results connect quantum implication with the operator theory of two projections and binary quantum discrimination.
\end{abstract}


\section{Introduction}
\label{sec:introduction}

Quantum logic represents experimental propositions by projections on a Hilbert space and organizes them by the lattice operations of closed subspaces \cite{birkhoff-von-neumann}. A persistent issue in this framework is the choice of implication because the classical material form need not internalize the projection order once distributivity fails. For example, if $P$ and $Q$ are distinct nonorthogonal line projections on $\mathbb C^2$, then $P^\perp\vee Q=I$ although $P\nleq Q$. Thus the classical lattice expression for material implication fails the entailment condition, which requires $P\Rightarrow Q=I$ exactly when $P\le Q$.

Polynomial implications on orthomodular lattices have been studied extensively. Herman, Marsden, and Piziak developed a systematic analysis of implication operations \cite{herman-marsden-piziak}, and Hardegree characterized the three polynomial material implications satisfying his four minimal implicative conditions \cite{hardegree}. More recent work has enlarged the class of quantum implications and has examined their logical and operational roles \cite{ozawa2017,ozawa2021,ozawa2026}. Rehder used operators derived from $PQP$ to define projection-valued connectives, including a material quasi-implication that coincides with the Sasaki implication \cite{rehder}. Younes and Schmitt studied a scalar-valued approach to quantum implication based on angles \cite{younes-schmitt}.

Throughout the introduction, $P$ and $Q$ denote projections, and $E^A(\Delta)$ denotes the spectral projection of a self-adjoint operator $A$ associated with a Borel set $\Delta$. Formal notation is fixed in Section~\ref{sec:preliminaries}. For projections, $P\le Q$ is equivalent to $Q-P\ge0$. Accordingly, we use $Q-P$ as the self-adjoint comparison operator. Passing from positivity to the nonnegative spectral projection, we obtain a projection-valued comparison for arbitrary pairs. The spectral and geometric properties of $Q-P$ belong to the established theory of differences of projections \cite{andruchow2014,bottcher-spitkovsky2010}. We call $E^{Q-P}([0,\infty))$ the \emph{truth projection} of the \emph{spectral sign implication}, which is formally defined in Section~\ref{sec:basic}. When $P$ and $Q$ commute, this value is the classical material implication. To our knowledge, this use of the spectral projection as a quantum logical connective has not previously been studied.

The main contribution of this paper is a characterization of this implication in terms of logical conditions and regularity. The operation satisfies Hardegree's four minimal implicative conditions, his law of contraposition, and a \emph{falsity condition}. It is not an ortholattice polynomial, even for two distinct nonorthogonal lines in $\mathbb C^2$. The standard decomposition for two projections represents operators on the generic summand by measurable $2\times2$ matrix fields over the spectrum of a central operator. On each \emph{generic fiber} of the uniform model, entailment and the falsity condition force a nontrivial truth projection, while contraposition restricts that projection to one of the two spectral branches of $Q-P$. Borel dependence permits measurable branch choices. Continuity prevents branch switching, and agreement with classical implication on commuting projections selects the positive branch.

This mechanism yields a complete classification within the uniform class of \emph{bounded Borel symbols}. The Borel operations satisfying entailment, contraposition, and the falsity condition correspond to Borel transversals, namely Borel sets $Z\subseteq[-1,1]$ such that $0,1\in Z$, $-1\notin Z$, and exactly one of $t$ and $-t$ belongs to $Z$ for every $0<t<1$. For continuous symbols satisfying these three conditions, only the positive spectral choice remains. Within the universal $C^*$-algebra generated by two projections, the same hypotheses also force the defining self-adjoint element to be a continuous function of $q-p$. The main results are Theorem~\ref{thm:branch-reduction}, Theorem~\ref{thm:symbol-classification}, Theorem~\ref{thm:continuous-uniqueness}, and Theorem~\ref{thm:observable-reduction}.

We also obtain an independent variational characterization from the operator difference. In finite dimensions, the implication is the largest projection maximizing $\operatorname{Tr}(R(Q-P))$. Using the standard Helstrom formulation of binary discrimination \cite{helstrom,holevo}, we show that, with normalized subspace states and priors proportional to rank, this projection is the largest among the acceptance projections of optimal projective tests. The De Morgan dual has a complementary extremal interpretation. On each principal angle plane of a generic pair, the De Morgan dual agrees with the midpoint projection known from the geometry of pairs of projections \cite{andruchow2014,davis}.

This paper is organized as follows. Section~\ref{sec:preliminaries} collects the standard facts about pairs of projections used throughout the subsequent sections. Section~\ref{sec:basic} examines the logical properties, qubit structure, structural limitations, and De Morgan geometry of the spectral sign implication. Section~\ref{sec:uniform} introduces the uniform construction framework and derives the branch reduction. Section~\ref{sec:classification} presents the Borel and continuous classifications, the observable reduction, and the independence results. Section~\ref{sec:variational} establishes the variational characterization and relates it to the Helstrom problem. Section~\ref{sec:qst} describes the connection with quantum set theory. Section~\ref{sec:conclusion} summarizes the main results and indicates directions for further study.

\section{Preliminaries}
\label{sec:preliminaries}

Throughout this paper, we let $\mathcal H$ be a complex Hilbert space and $\mathcal M$ a von Neumann algebra on $\mathcal H$, with $\operatorname{Proj}(\mathcal M)$ its projection lattice. We use capital letters $P,Q,R$ for projections and $A$ for a general self-adjoint operator, and reserve $B$ and $S$ for the comparison operators defined below. For a self-adjoint operator $A$ and a Borel set $\Delta\subseteq\mathbb R$, we write $E^A(\Delta)$ for the associated spectral projection. We take $I$ as the identity projection, set $P^\perp=I-P$, and use $M_2(\mathbb C)$ and $M_2(\mathbb C)_{\mathrm{sa}}$ for the algebra of $2\times2$ complex matrices and its self-adjoint part, respectively. For operators $T_1,\ldots,T_n$, we write $W^*(T_1,\ldots,T_n)$ for the von Neumann algebra they generate.

\begin{definition}[Generic position]
\label{def:generic-position}
Two projections $P$ and $Q$ are in \emph{generic position} if $P\wedge Q=P\wedge Q^\perp=P^\perp\wedge Q=P^\perp\wedge Q^\perp=0$.
\end{definition}

The terminology is due to Halmos \cite{halmos}. Every pair of projections admits an orthogonal decomposition into four common eigenspaces and a generic summand. For $i,j\in\{0,1\}$, let $\mathcal H_{ij}$ be the subspace on which $P=iI$ and $Q=jI$. Thus $\mathcal H_{11}=\operatorname{Ran}P\cap\operatorname{Ran}Q$, $\mathcal H_{10}=\operatorname{Ran}P\cap\ker Q$, $\mathcal H_{01}=\ker P\cap\operatorname{Ran}Q$, and $\mathcal H_{00}=\ker P\cap\ker Q$. On each of these four subspaces, both projections act as scalars and therefore commute. We denote the orthogonal complement of their direct sum by $\mathcal H_g$. The restrictions of $P$ and $Q$ to $\mathcal H_g$ are in generic position.

The next formula separates the three exceptional spectral values that occur on the commuting sectors.

\begin{lemma}[Exceptional spectral sectors]
\label{lem:exceptional-sectors}
For all projections $P,Q$, one has $E^{Q-P}(\{1\})=P^\perp\wedge Q$, $E^{Q-P}(\{-1\})=P\wedge Q^\perp$, and $E^{Q-P}(\{0\})=(P\wedge Q)\vee(P^\perp\wedge Q^\perp)$.
\end{lemma}

\begin{proof}
Suppose $(Q-P)x=x$. Then $\langle Qx,x\rangle-\langle Px,x\rangle=\|x\|^2$. Since $0\le\langle Qx,x\rangle\le\|x\|^2$ and $\langle Px,x\rangle\ge0$, equality forces $Qx=x$ and $Px=0$. The converse is immediate. Hence the first formula holds. The second follows by interchanging $P$ and $Q$ and using $E^{P-Q}(\{1\})=E^{Q-P}(\{-1\})$.

If $(Q-P)x=0$, put $y=Px=Qx$. Then $Py=y$ and $Qy=y$. Moreover, $P(x-y)=Q(x-y)=0$. Thus $x=y+(x-y)$ belongs to $(P\wedge Q)\mathcal H\oplus(P^\perp\wedge Q^\perp)\mathcal H$. The converse inclusion is immediate.
\end{proof}

The zero spectral sector gives the kernel formula used below.

\begin{lemma}[Kernel formula]
\label{lem:kernel}
For all projections $P,Q$, $\ker(Q-P)=(P\wedge Q)\mathcal H\oplus(P^\perp\wedge Q^\perp)\mathcal H$.
\end{lemma}

\begin{proof}
This is the statement about the zero spectral sector in Lemma~\ref{lem:exceptional-sectors}.
\end{proof}

The following elementary operators will be used throughout this paper.

\begin{definition}[Comparison operators]
\label{def:comparison-operators}
Let $P,Q$ be projections and put $B=Q-P$ and $S=P+Q-I$.
\end{definition}

These comparison operators are standard in the operator theory of pairs of projections \cite{avron-seiler-simon}.

\begin{lemma}[Projection pair identities]
\label{lem:two-projection-identities}
The operators in Definition~\ref{def:comparison-operators} satisfy $BS=-SB$ and $B^2+S^2=I$. Both $B^2$ and $S^2$ commute with $P$ and $Q$. On the generic summand, $\ker B=\ker S=0$.
\end{lemma}

\begin{proof}
Multiplying directly, we obtain $BS=(Q-P)(P+Q-I)=QP-PQ$ and $SB=PQ-QP$, so $BS=-SB$. Also, $B^2=P+Q-PQ-QP$ and $S^2=I-P-Q+PQ+QP$, hence $B^2+S^2=I$. Moreover, $PB^2=B^2P=P-PQP$. The same statement for $Q$ follows symmetrically. The commutation of $S^2=I-B^2$ follows as well. By Lemma~\ref{lem:kernel}, $\ker B=(P\wedge Q)\mathcal H\oplus(P^\perp\wedge Q^\perp)\mathcal H$, so $B$ is injective on the generic summand. The corresponding statement for $S=P-Q^\perp$ follows by replacing $Q$ with $Q^\perp$ in the same formula.
\end{proof}

The anticommutation relation gives a spectral symmetry on the generic summand. The lemma below is stated in terms of spectral projections, so it also applies when the spectrum has a continuous part.

\begin{lemma}[Spectral symmetry]
\label{lem:spectral-symmetry}
Let $P,Q$ be in generic position and let $B,S$ be as in Definition~\ref{def:comparison-operators}. Put $U=\operatorname{sgn}(S)$, where $\operatorname{sgn}$ denotes the usual sign function and is applied to $S$ by Borel functional calculus. Then $U$ is a self-adjoint unitary, $UBU=-B$, and $U E^B(\Delta)U=E^B(-\Delta)$ for every Borel set $\Delta\subseteq\mathbb R$.
\end{lemma}

\begin{proof}
By Lemma~\ref{lem:two-projection-identities}, $BS=-SB$, $B^2S=SB^2$, and $BS^2=S^2B$. Hence $B$ commutes with $|S|=(S^2)^{1/2}$. On the generic summand, $\ker S=0$. The polar decomposition is therefore $S=U|S|$ with $U$ a self-adjoint unitary. From $BU|S|=-UB|S|$ and the density of $\operatorname{Ran}|S|$, we obtain $BU=-UB$. Hence $UBU=-B$. By unitary covariance of the Borel functional calculus, $U E^B(\Delta)U=E^{UBU}(\Delta)=E^{-B}(\Delta)=E^B(-\Delta)$.
\end{proof}

\section{Spectral Sign Implication}
\label{sec:basic}

We begin with the definition and basic laws of the spectral sign implication. The qubit calculation comes next and is used in the subsequent analysis of structural limitations and the De Morgan dual.

\subsection{Basic Properties}
\label{subsec:basic-properties}

The starting point is the nonnegative spectral sector of the comparison operator.

\begin{definition}[Spectral sign implication]
\label{def:spectral-implication}
Let $P,Q\in\operatorname{Proj}(\mathcal M)$. The \emph{spectral sign implication} is
\[
P\Rightarrow_{\mathrm{sp}}Q=E^{Q-P}([0,\infty)).
\]
\end{definition}

The closed threshold includes the zero sector. The next proposition shows that this choice is forced if the operation is to agree with classical material implication on commuting projections.

\begin{proposition}[Boolean reduction]
\label{prop:boolean-reduction}
If $P$ and $Q$ commute, then $P\Rightarrow_{\mathrm{sp}}Q=P^\perp\vee Q$.
\end{proposition}

\begin{proof}
The value of $Q-P$ is $0$ on $\mathcal H_{11}$ and $\mathcal H_{00}$. Its values on $\mathcal H_{10}$ and $\mathcal H_{01}$ are $-I$ and $I$, respectively. The nonnegative spectral projection therefore contains every sector except $\mathcal H_{10}$. This is exactly $P^\perp\vee Q$.
\end{proof}

The operation is local in the usual sense for a pair of projections. We state both parts explicitly.

\begin{proposition}[Locality]
\label{prop:locality}
The projection $P\Rightarrow_{\mathrm{sp}}Q$ belongs to $W^*(P,Q)$. If a projection $R$ commutes with $P$ and $Q$, then $(P\Rightarrow_{\mathrm{sp}}Q)\wedge R=((P\wedge R)\Rightarrow_{\mathrm{sp}}(Q\wedge R))\wedge R$.
\end{proposition}

\begin{proof}
The first statement follows from Borel functional calculus because $Q-P\in W^*(P,Q)$. For the second statement, decompose $\mathcal H=R\mathcal H\oplus R^\perp\mathcal H$. The operators $P$ and $Q$ reduce this decomposition. On $R\mathcal H$, the difference of $P\wedge R$ and $Q\wedge R$ is the restriction of $Q-P$. On $R^\perp\mathcal H$, both restricted projections vanish. Taking the nonnegative spectral projection and then meeting with $R$, we obtain the asserted equality.
\end{proof}

The operation is also natural under normal morphisms of von Neumann algebras.

\begin{proposition}[Normal covariance]
\label{prop:normal-covariance}
Let $\mathcal N$ be a von Neumann algebra and let $\pi:\mathcal M\to\mathcal N$ be a normal unital $*$-homomorphism. Then $\pi(P\Rightarrow_{\mathrm{sp}}Q)=\pi(P)\Rightarrow_{\mathrm{sp}}\pi(Q)$.
\end{proposition}

\begin{proof}
Put $A=Q-P$. Let $f_n:\mathbb R\to[0,1]$ be continuous for $n\ge1$, equal to $1$ on $(-\infty,-1/n]$, equal to $0$ on $[0,\infty)$, and linear on $[-1/n,0]$. Then $f_n(A)$ increases strongly to $E^A(({-\infty},0))$. By continuous functional calculus, $\pi(f_n(A))=f_n(\pi(A))$. Since $\pi$ is normal, it preserves the supremum of this bounded increasing sequence. Hence $\pi(E^A(({-\infty},0)))=E^{\pi(A)}(({-\infty},0))$. Taking orthocomplements and using unitality, we obtain $\pi(E^A([0,\infty)))=E^{\pi(A)}([0,\infty))$. Since $\pi(A)=\pi(Q)-\pi(P)$, the asserted covariance follows.
\end{proof}

We consider Hardegree's four minimal implicative conditions \cite{hardegree}, together with his law of contraposition (CN) and a condition on the bottom value. These conditions are stated explicitly in the following definition.

\begin{definition}[Implication conditions]
\label{def:implication-conditions}
A binary operation $\Rightarrow$ on projections satisfies the following conditions when the stated relations hold for all projections $P,Q$.
\begin{caselist}
\item \emph{Entailment} (E): $P\Rightarrow Q=I$ exactly when $P\le Q$.
\item \emph{Modus ponens} (MP): $P\wedge(P\Rightarrow Q)\le Q$.
\item \emph{Modus tollens} (MT): $Q^\perp\wedge(P\Rightarrow Q)\le P^\perp$.
\item The \emph{negation condition} (NG): $P\wedge Q^\perp\le(P\Rightarrow Q)^\perp$.
\item \emph{Contraposition} (CN): $P\Rightarrow Q=Q^\perp\Rightarrow P^\perp$.
\item The \emph{falsity condition} (F): $P\Rightarrow Q=0$ exactly when $P=I$ and $Q=0$.
\end{caselist}
\end{definition}

For the uniform constructions considered later, a condition is understood to hold for every pair of projections on every complex Hilbert space.

The falsity condition requires the operation to take the value $0$ exactly in the same case as classical material implication. It is not automatically satisfied by the standard polynomial quantum implications. Hardegree's three polynomial material implications satisfying all four minimal implicative conditions are the Sasaki implication, the contrapositive Sasaki implication, and the relevance implication. All three satisfy his nonstrictness condition (NS), namely $I\Rightarrow Q=Q$ and $P\Rightarrow0=P^\perp$ \cite{hardegree}. The first two satisfy (F), while the relevance implication can take the value $0$ on a generic qubit pair. Hence (F) does not follow from Hardegree's four minimal implicative conditions even when (NS) is added. The comparison in Subsection~\ref{subsec:qubit} gives an explicit example.

The spectral sign implication satisfies modus ponens and modus tollens in stronger equality forms.

\begin{theorem}[Basic implication laws]
\label{thm:basic-laws}
The spectral sign implication satisfies (E), (MP), (MT), (NG), and (F). More precisely,
\[
P\wedge(P\Rightarrow_{\mathrm{sp}}Q)=P\wedge Q,
\qquad
Q^\perp\wedge(P\Rightarrow_{\mathrm{sp}}Q)=P^\perp\wedge Q^\perp.
\]
\end{theorem}

\begin{proof}
For (E), $P\Rightarrow_{\mathrm{sp}}Q=I$ holds exactly when $Q-P\ge0$. Since $P$ and $Q$ are projections, $Q-P\ge0$ is equivalent to $P\le Q$.

For (MP), put $R=P\Rightarrow_{\mathrm{sp}}Q$. If $x\in P\mathcal H\cap R\mathcal H$, then $\langle x,(Q-P)x\rangle\ge0$. Since $Px=x$, this becomes $\langle Qx,x\rangle\ge\|x\|^2$. For a projection, $\langle Qx,x\rangle=\|Qx\|^2\le\|x\|^2$. Thus $Qx=x$, and hence $P\wedge R\le P\wedge Q$. The reverse inclusion lies in $\ker(Q-P)$ and hence in $R$. Therefore the modus ponens equality holds.

For (MT), if $x\in Q^\perp\mathcal H\cap R\mathcal H$, then $\langle x,(Q-P)x\rangle\ge0$ becomes $-\langle Px,x\rangle\ge0$. Hence $Px=0$. The reverse inclusion again lies in the kernel of $Q-P$. Therefore the modus tollens equality holds.

For (NG), Lemma~\ref{lem:exceptional-sectors} gives $P\wedge Q^\perp=E^{Q-P}(\{-1\})$. This projection is dominated by $E^{Q-P}(({-\infty},0))=R^\perp$.

For (F), suppose $E^{Q-P}([0,\infty))=0$. By the spectral theorem, $\langle x,(Q-P)x\rangle<0$ for every nonzero $x$. If $0\ne x\in\ker P$, however, then $\langle x,(Q-P)x\rangle=\|Qx\|^2\ge0$, a contradiction. Hence $\ker P=0$, so $P=I$. If $Q\ne0$, choose $0\ne x\in\operatorname{Ran}Q$. Then $(Q-I)x=0$, so the nonnegative spectral subspace of $Q-I$ is nonzero, contradicting $E^{Q-I}([0,\infty))=0$. Hence $Q=0$. The converse follows from $E^{-I}([0,\infty))=0$.
\end{proof}

Contraposition follows directly from the form of the comparison operator.

\begin{proposition}[Contraposition]
\label{prop:contraposition}
For all projections $P,Q$, $P\Rightarrow_{\mathrm{sp}}Q=Q^\perp\Rightarrow_{\mathrm{sp}}P^\perp$.
\end{proposition}

\begin{proof}
The comparison operator on the right is $P^\perp-Q^\perp=Q-P$.
\end{proof}

The values at the top and bottom projections are also explicit.

\begin{proposition}[Boundary laws]
\label{prop:boundary-laws}
The identities $0\Rightarrow_{\mathrm{sp}}Q=I$, $P\Rightarrow_{\mathrm{sp}}I=I$, $I\Rightarrow_{\mathrm{sp}}Q=Q$, and $P\Rightarrow_{\mathrm{sp}}0=P^\perp$ hold for all projections $P,Q$.
\end{proposition}

\begin{proof}
The four comparison operators are $Q$, $I-P$, $Q-I=-Q^\perp$, and $-P$. Their nonnegative spectral projections are precisely the four projections in the statement.
\end{proof}

The last two identities are Hardegree's nonstrictness condition (NS) \cite{hardegree}. Thus the spectral sign implication is nonstrict in his terminology. Hardegree notes that the polynomial condition (P) subsumes (NS). By Theorem~\ref{thm:nonpolynomial}, the spectral sign implication satisfies (NS) without satisfying (P).

The meet of the implication in both directions gives the projection onto the subspace on which the two projections agree.

\begin{proposition}[Logical equivalence]
\label{prop:logical-equivalence}
For all projections $P,Q$,
\[
(P\Rightarrow_{\mathrm{sp}}Q)\wedge(Q\Rightarrow_{\mathrm{sp}}P)
=(P\wedge Q)\vee(P^\perp\wedge Q^\perp).
\]
\end{proposition}

\begin{proof}
Both factors are spectral projections of $B=Q-P$. The first is $E^B([0,\infty))$ and the second is $E^B(({-\infty},0])$. Their meet is $E^B(\{0\})$. Apply Lemma~\ref{lem:exceptional-sectors}.
\end{proof}

\subsection{Qubit Structure}
\label{subsec:qubit}

The case of $\mathbb C^2$ admits a useful geometric formula. This formula also shows that the operation is not an ortholattice polynomial.

\begin{definition}[Line projections]
\label{def:line-projections}
Let $\theta\in\mathbb R$. In $\mathbb C^2$, let $P_\theta$ denote the projection onto the complex line spanned by $(\cos\theta,\sin\theta)$. Angles are understood modulo $\pi$.
\end{definition}

The spectral projection can now be computed directly for two distinct lines.

\begin{proposition}[Signed angle formula]
\label{prop:signed-angle}
Let $\theta,\eta\in\mathbb R$, and choose $\delta=\eta-\theta$ with $-\pi/2<\delta<\pi/2$ and $\delta\ne0$. Then
\[
P_\theta\Rightarrow_{\mathrm{sp}}P_\eta
=P_{(\theta+\eta)/2+(\pi/4)\operatorname{sgn}(\delta)}.
\]
\end{proposition}

\begin{proof}
Conjugation by the real rotation through angle $-\theta$ reduces the calculation to $\theta=0$ and $\eta=\delta$. Write $c=\cos\delta$ and $s=\sin\delta$. Then $P_\delta-P_0=\begin{pmatrix}-s^2&cs\\cs&s^2\end{pmatrix}$ and its square is $s^2I$. Substituting directly, we find that when $s>0$, the unit vector at angle $\pi/4+\delta/2$ is an eigenvector with eigenvalue $s$, while when $s<0$, the unit vector at angle $-\pi/4+\delta/2$ is an eigenvector with eigenvalue $-s$. In either case, this vector spans the positive eigenspace. After rotating back through angle $\theta$, we obtain the formula.
\end{proof}

The signed angle formula immediately gives the obstruction to polynomial representation.

\begin{theorem}[Nonpolynomial behavior]
\label{thm:nonpolynomial}
Let $P=P_0$ and $Q=P_\alpha$ with $0<\alpha<\pi/2$. Then $P\Rightarrow_{\mathrm{sp}}Q=P_{\pi/4+\alpha/2}$ does not belong to the ortholattice generated by $P$ and $Q$. Hence the spectral sign implication is not an ortholattice term operation.
\end{theorem}

\begin{proof}
Two distinct nonorthogonal lines in $\mathbb C^2$ generate the ortholattice with six elements $\{0,I,P,P^\perp,Q,Q^\perp\}$. By Proposition~\ref{prop:signed-angle}, the resulting line has angle $\pi/4+\alpha/2$. When $0<\alpha<\pi/2$, this angle is distinct modulo $\pi$ from $0$, $\pi/2$, $\alpha$, and $\alpha+\pi/2$. The resulting projection therefore lies outside the generated ortholattice.
\end{proof}

For comparison, Hardegree's three polynomial material implications are the Sasaki implication $P^\perp\vee(P\wedge Q)$, the contrapositive Sasaki implication $(P\vee Q)^\perp\vee Q$, and the relevance implication $(P\wedge Q)\vee(P^\perp\wedge Q)\vee(P^\perp\wedge Q^\perp)$ \cite{hardegree}. When $P=P_0$ and $Q=P_{\pi/6}$, these three implications have values $P_{\pi/2}$, $P_{\pi/6}$, and $0$, respectively, whereas the spectral sign implication has value $P_{\pi/3}$ by Proposition~\ref{prop:signed-angle}. This example separates the present operation from all three polynomial material implications.

\subsection{Structural Limitations}
\label{subsec:limits}

The following results describe some logical limitations of the spectral sign implication and are independent of the characterization developed later. The implication laws in Theorem~\ref{thm:basic-laws} do not imply the order monotonicities associated with residuated implications, and the operation also fails Hardegree's transitivity condition. These failures already occur on $\mathbb C^2$.

\begin{proposition}[Consequent monotonicity failure]
\label{prop:consequent-failure}
There are projections $P,Q,R$ with $Q\le R$ and $P\Rightarrow_{\mathrm{sp}}Q\nleq P\Rightarrow_{\mathrm{sp}}R$.
\end{proposition}

\begin{proof}
Take $P=P_0$, $Q=0$, and $R=P_\alpha$ with $0<\alpha<\pi/2$. Then $P\Rightarrow_{\mathrm{sp}}Q=P_{\pi/2}$, while Proposition~\ref{prop:signed-angle} yields $P\Rightarrow_{\mathrm{sp}}R=P_{\pi/4+\alpha/2}$. The two projections are onto different lines, so neither is below the other.
\end{proof}

The failure of consequent monotonicity has an immediate consequence for residuation.

\begin{corollary}[Residuation obstruction]
\label{cor:residuation}
On $\operatorname{Proj}(M_2(\mathbb C))$, the spectral sign implication cannot be the right residual of any binary operation. More precisely, there is no binary operation $\otimes$ on $\operatorname{Proj}(M_2(\mathbb C))$ such that $P\otimes R\le Q$ exactly when $R\le P\Rightarrow_{\mathrm{sp}}Q$ for all projections $P,Q,R$.
\end{corollary}

\begin{proof}
With $P$ fixed, a right adjoint $Q\mapsto P\Rightarrow Q$ is monotone. Proposition~\ref{prop:consequent-failure} rules out this property.
\end{proof}

The logical role of the spectral sign implication therefore differs in general from that of implications based on residuation.

Hardegree's residual condition (R) requires a binary operation $+$ such that $a+b\le c$ exactly when $a\le b\to c$ \cite{hardegree}. Setting $a+b=b\otimes a$ converts (R) into the residuation property ruled out by Corollary~\ref{cor:residuation}, so (R) fails on $\operatorname{Proj}(M_2(\mathbb C))$. The spectral sign implication satisfies the other relevant conditions (E), (MP), and (NG) on every projection lattice. Hardegree's representation theorem \cite[Theorem~2]{hardegree} therefore does not provide a general representation of the present operation.

\begin{proposition}[Antecedent antitonicity failure]
\label{prop:antecedent-failure}
There are projections $P,Q,R$ with $P\le R$ and $R\Rightarrow_{\mathrm{sp}}Q\nleq P\Rightarrow_{\mathrm{sp}}Q$.
\end{proposition}

\begin{proof}
Take $P=P_0$, $R=I$, and $Q=P_\alpha$ with $0<\alpha<\pi/2$. By Proposition~\ref{prop:boundary-laws}, $R\Rightarrow_{\mathrm{sp}}Q=P_\alpha$, whereas Proposition~\ref{prop:signed-angle} yields $P\Rightarrow_{\mathrm{sp}}Q=P_{\pi/4+\alpha/2}$. These are projections onto different lines, so the required order relation fails.
\end{proof}

Hardegree's transitivity condition (T) requires $(P\Rightarrow Q)\wedge(Q\Rightarrow R)\le P\Rightarrow R$ \cite{hardegree}. Under (E), condition (T) implies antitonicity in the antecedent.

\begin{corollary}[Transitivity failure]
\label{cor:transitivity}
There are projections $P,Q,R$ with $(P\Rightarrow_{\mathrm{sp}}Q)\wedge(Q\Rightarrow_{\mathrm{sp}}R)\nleq P\Rightarrow_{\mathrm{sp}}R$.
\end{corollary}

\begin{proof}
By Proposition~\ref{prop:antecedent-failure}, there are projections $P\le Q$ and $R$ with $Q\Rightarrow_{\mathrm{sp}}R\nleq P\Rightarrow_{\mathrm{sp}}R$. Condition (E) gives $P\Rightarrow_{\mathrm{sp}}Q=I$. Hence $(P\Rightarrow_{\mathrm{sp}}Q)\wedge(Q\Rightarrow_{\mathrm{sp}}R)=Q\Rightarrow_{\mathrm{sp}}R\nleq P\Rightarrow_{\mathrm{sp}}R$.
\end{proof}

Hardegree proved that an orthomodular lattice admitting an implication satisfying (E), (MP), (T), and (NS) is Boolean \cite[Theorem~1]{hardegree}. Since the spectral sign implication satisfies (E), (MP), and (NS), condition (T) fails on every non-Boolean projection lattice.

The spectral threshold also causes norm discontinuity at the diagonal. The continuity considered later concerns the defining matrix symbol, whereas the following statement concerns norm continuity of the resulting binary operation in its projection arguments.

\begin{proposition}[Diagonal discontinuity]
\label{prop:diagonal-discontinuity}
For every projection $P_\theta$ onto a line in $\mathbb C^2$, the map $(X,Y)\mapsto X\Rightarrow_{\mathrm{sp}}Y$ is not norm continuous at $(P_\theta,P_\theta)$.
\end{proposition}

\begin{proof}
Fix a projection $P_\theta$ onto a line and let $\eta\to\theta$ through values with $|\eta-\theta|<\pi/2$. By Proposition~\ref{prop:signed-angle}, the limits of $P_\theta\Rightarrow_{\mathrm{sp}}P_\eta$ from the two sides are $P_{\theta+\pi/4}$ and $P_{\theta-\pi/4}$. By Proposition~\ref{prop:boundary-laws}, $P_\theta\Rightarrow_{\mathrm{sp}}P_\theta=I$. Neither limit equals $I$.
\end{proof}

\subsection{De Morgan Geometry}
\label{subsec:dual}

The De Morgan dual of the spectral sign implication uses the opposite convention at the zero spectral value.

\begin{definition}[Spectral dual conjunction]
\label{def:dual-conjunction}
Let $P,Q$ be projections. Define the \emph{spectral dual conjunction} by $P*_{\mathrm{sp}}Q=(P\Rightarrow_{\mathrm{sp}}Q^\perp)^\perp$.
\end{definition}

The next proposition expresses the dual directly in terms of the second comparison operator.

\begin{proposition}[Spectral formula]
\label{prop:dual-formula}
For all projections $P,Q$, one has $P*_{\mathrm{sp}}Q=E^{P+Q-I}((0,\infty))$. If $P$ and $Q$ commute, then $P*_{\mathrm{sp}}Q=P\wedge Q$.
\end{proposition}

\begin{proof}
By Definition~\ref{def:spectral-implication}, $P\Rightarrow_{\mathrm{sp}}Q^\perp=E^{I-P-Q}([0,\infty))$. Taking the orthocomplement produces $E^{I-P-Q}(({-\infty},0))$, which equals $E^{P+Q-I}((0,\infty))$. If $P$ and $Q$ commute, the operator $P+Q-I$ has value $1$ on $\mathcal H_{11}$ and value $-1$ on $\mathcal H_{00}$. It has value $0$ on both $\mathcal H_{10}$ and $\mathcal H_{01}$. Hence its positive spectral projection is $P\wedge Q$.
\end{proof}

The open threshold in Proposition~\ref{prop:dual-formula} is forced by the Boolean reduction. Replacing $(0,\infty)$ by $[0,\infty)$, we include the two Boolean zero sectors and obtain $P\vee Q$ in the commuting case. For the implication, replacing the closed threshold $[0,\infty)$ by $(0,\infty)$, we obtain $P^\perp\wedge Q$ instead of material implication.

\begin{proposition}[Kernel correction]
\label{prop:kernel-correction}
For all projections $P,Q$,
\[
P\Rightarrow_{\mathrm{sp}}Q
=(P^\perp*_{\mathrm{sp}}Q)\vee E^{Q-P}(\{0\}).
\]
Equivalently,
\[
P\Rightarrow_{\mathrm{sp}}Q
=(P^\perp*_{\mathrm{sp}}Q)\vee(P\wedge Q)\vee(P^\perp\wedge Q^\perp).
\]
If $P,Q$ are in generic position, then $P\Rightarrow_{\mathrm{sp}}Q=P^\perp*_{\mathrm{sp}}Q$.
\end{proposition}

\begin{proof}
By Proposition~\ref{prop:dual-formula}, $P^\perp*_{\mathrm{sp}}Q=E^{Q-P}((0,\infty))$. The nonnegative spectral projection is the orthogonal join of this projection and $E^{Q-P}(\{0\})$. The second formula follows from Lemma~\ref{lem:exceptional-sectors}. If $P$ and $Q$ are in generic position, then $Q-P$ has zero kernel by Definition~\ref{def:generic-position} and Lemma~\ref{lem:kernel}.
\end{proof}

From the generic qubit calculation, we obtain a midpoint formula. The corresponding result in finite dimensions follows from the principal angle decomposition of two subspaces.

\begin{proposition}[Principal angle midpoint]
\label{prop:midpoint}
Let $P=P_0$ and $Q=P_\alpha$ with $0<\alpha<\pi/2$. Then $P*_{\mathrm{sp}}Q=P_{\alpha/2}$. More generally, let $P,Q$ be projections of finite rank in generic position with equal rank and principal angles $\theta_1,\ldots,\theta_k$. On each principal angle plane, $P*_{\mathrm{sp}}Q$ is the projection onto the line at half the corresponding principal angle.
\end{proposition}

\begin{proof}
In the qubit case, $P+Q-I=\begin{pmatrix}\cos^2\alpha&\sin\alpha\cos\alpha\\\sin\alpha\cos\alpha&-\cos^2\alpha\end{pmatrix}$. Its positive eigenspace is the line at angle $\alpha/2$.

In the general case, put $\operatorname{rank}P=\operatorname{rank}Q=k$. Since $P$ and $Q$ are in generic position, $\operatorname{Ran}P\cap\operatorname{Ran}Q=\{0\}$ and $\ker P\cap\ker Q=\{0\}$. Since $(\operatorname{Ran}P+\operatorname{Ran}Q)^\perp=\ker P\cap\ker Q=\{0\}$, the sum $\operatorname{Ran}P+\operatorname{Ran}Q$ is dense in $\mathcal H$. This sum is finite dimensional and therefore closed, so $\operatorname{Ran}P+\operatorname{Ran}Q=\mathcal H$ and $\dim\mathcal H=2k$. By the principal angle decomposition \cite{amrein-sinha,davis}, there is an orthogonal decomposition $\mathcal H=\bigoplus_{j=1}^k\mathcal K_j$ into planes that reduce $P$ and $Q$. In suitable orthonormal coordinates on $\mathcal K_j$, the restrictions are $P_0$ and $P_{\theta_j}$, where $0<\theta_j<\pi/2$ because generic position excludes the cases $\theta_j=0$ and $\theta_j=\pi/2$. The operator $P+Q-I$ and its spectral projections reduce the same orthogonal decomposition. Applying the qubit calculation on each $\mathcal K_j$, we obtain the projection onto the line at angle $\theta_j/2$ as the positive spectral projection on that plane.
\end{proof}

This midpoint interpretation concerns the dual conjunction. Proposition~\ref{prop:kernel-correction} describes the exact relation between that geometry and the implication.

\begin{proposition}[Nonassociativity]
\label{prop:nonassociativity}
The operation $*_{\mathrm{sp}}$ is not associative.
\end{proposition}

\begin{proof}
The spectral formula in Proposition~\ref{prop:dual-formula} is covariant under simultaneous unitary conjugation. Hence the midpoint formula in Proposition~\ref{prop:midpoint} remains valid after a common rotation. We have $P_0*_{\mathrm{sp}}P_{\pi/6}=P_{\pi/12}$ and $P_{\pi/6}*_{\mathrm{sp}}P_{\pi/3}=P_{\pi/4}$. Applying the rotated midpoint formula once more, we obtain $(P_0*_{\mathrm{sp}}P_{\pi/6})*_{\mathrm{sp}}P_{\pi/3}=P_{5\pi/24}$ and $P_0*_{\mathrm{sp}}(P_{\pi/6}*_{\mathrm{sp}}P_{\pi/3})=P_{\pi/8}$. These projections are distinct.
\end{proof}

\section{Uniform Constructions}
\label{sec:uniform}

The characterization requires a uniform framework for constructions from arbitrary pairs of projections. This setting allows the argument to pass from global operations to generic fibers.

\subsection{Symbol Class}
\label{subsec:symbols}

The standard theory of two projections provides both continuous and measurable models for the required uniform constructions.

\begin{definition}[Canonical fibers]
\label{def:canonical-fibers}
Let $x\in[0,1]$. Define
\[
p_x=\begin{pmatrix}1&0\\0&0\end{pmatrix},\qquad
q_x=\begin{pmatrix}x&\sqrt{x(1-x)}\\\sqrt{x(1-x)}&1-x\end{pmatrix}.
\]
The pair $(p_x,q_x)$ is the \emph{canonical fiber} at $x$. For $0<x<1$, this canonical fiber is in generic position and will be called a \emph{generic fiber}. Put $b_x=q_x-p_x$ and $s_x=p_x+q_x-I$.
\end{definition}

Direct calculation gives $b_x^2=(1-x)I$ and $s_x^2=xI$. When $0<x<1$, the pair $p_x,q_x$ generates $M_2(\mathbb C)$. The two eigenvalues of $b_x$ are $\pm\sqrt{1-x}$. Writing $x=\cos^2\theta$ with $0\le\theta\le\pi/2$, we have $p_x=P_0$ and $q_x=P_\theta$, while $\sqrt{1-x}=\sin\theta$. Thus the parameter $\sqrt{1-x}$ is the sine of the angle between the two canonical lines.

By the Raeburn--Sinclair model, the universal unital $C^*$-algebra $C^*(p,q)$ generated by two projections is isomorphic to the algebra of continuous functions $f:[0,1]\to M_2(\mathbb C)$ whose values at $0$ and $1$ are diagonal \cite{raeburn-sinclair}. In this representation, the universal generators are the functions $x\mapsto p_x$ and $x\mapsto q_x$ from Definition~\ref{def:canonical-fibers}. The diagonal endpoint requirement in the following definition is the corresponding boundary condition, while Borel measurability allows the same fiber description in the measurable setting.

\begin{definition}[Bounded Borel symbol for two projections]
\label{def:borel-symbol}
A \emph{bounded Borel symbol for two projections} is a bounded Borel map $\Phi:[0,1]\to M_2(\mathbb C)_{\mathrm{sa}}$ such that $\Phi(0)$ and $\Phi(1)$ are diagonal. Such a symbol is called a \emph{continuous symbol for two projections} when $\Phi$ is continuous.
\end{definition}

We next describe the uniform construction associated with such a symbol. On the generic summand of a pair $P,Q$, write the Halmos form with first coordinate $\operatorname{Ran}P$ as $P=\begin{pmatrix}I&0\\0&0\end{pmatrix}$ and $Q=\begin{pmatrix}H&K\\K&I-H\end{pmatrix}$, where $K=(H(I-H))^{1/2}$. Then $S^2=\operatorname{diag}(H,H)$, so $H$ is the representative of the intrinsic central operator $S^2$ on the multiplicity space. This convention agrees with Definition~\ref{def:canonical-fibers}, where $s_x^2=xI$. The algebra $W^*(P,Q)$ on the generic summand is a type $I_2$ algebra over $W^*(H)$ and is represented by measurable $2\times2$ matrix fields \cite{spitkovsky1994}.

\begin{definition}[Uniform symbol construction]
\label{def:uniform-symbol-construction}
Let $\Phi=(\phi_{ij})$ be a bounded Borel symbol for two projections, let $P,Q$ be projections, and let $H$ be the central parameter in the Halmos realization described above. Define the component of $A_\Phi(P,Q)$ on the generic summand by applying the Borel functional calculus of $H$ to each scalar entry. On $\mathcal H_{10},\mathcal H_{01},\mathcal H_{11},\mathcal H_{00}$, use the first entry of $\Phi(0)$, the second entry of $\Phi(0)$, the first entry of $\Phi(1)$, and the second entry of $\Phi(1)$, respectively. The first diagonal slot corresponds to $\operatorname{Ran}P$ and the second to $\ker P$.
\end{definition}

The next lemma shows that this construction is independent of the Halmos realization and covariant under simultaneous unitary conjugation.

\begin{lemma}[Symbol covariance]
\label{lem:symbol-covariance}
The operator $A_\Phi(P,Q)$ is independent of the chosen Halmos realization. It satisfies $A_\Phi(VPV^*,VQV^*)=V A_\Phi(P,Q)V^*$ for every unitary $V$.
\end{lemma}

\begin{proof}
We use the standard Halmos decomposition with multiplicity \cite{bottcher-spitkovsky2010,spitkovsky1994}. Consider two Halmos realizations of the generic summand, with parameters $H$ and $H'$. Put $K=(H(I-H))^{1/2}$ and $K'=(H'(I-H'))^{1/2}$ for their off-diagonal blocks. Let $W_1\oplus W_2$ be a unitary intertwining the two realizations. From the diagonal and off-diagonal blocks of $Q$, respectively, we obtain $W_1H=H'W_1$ and $W_1K=K'W_2$. The first relation implies $W_1K=K'W_1$ by continuous functional calculus. By Lemma~\ref{lem:two-projection-identities}, the operators corresponding to $Q-P$ and $P+Q-I$ in the second Halmos realization are injective. Their squares are $\operatorname{diag}(I-H',I-H')$ and $\operatorname{diag}(H',H')$, so $\ker(I-H')=\ker H'=0$. Hence $K'$ is injective, and therefore $W_1=W_2$. Write the common unitary as $W$. The relation $WH=H'W$ implies $Wf(H)=f(H')W$ for every bounded Borel function $f$. Hence $W\phi_{ij}(H)=\phi_{ij}(H')W$ for each entry $\phi_{ij}$, so $W\oplus W$ intertwines the two matrix fields obtained from $\Phi$. Hence $A_\Phi(P,Q)$ is independent of the realization. Simultaneous conjugation by $V$ carries the four intersection sectors, the generic summand, and the central parameter to the corresponding data for $VPV^*,VQV^*$. Applying the same symbol on each corresponding summand, we obtain the covariance formula.
\end{proof}

We now apply the fixed spectral threshold to this covariant self-adjoint construction.

\begin{definition}[Symbol implication]
\label{def:symbol-implication}
Let $\Phi$ be a bounded Borel symbol for two projections, and let $P,Q$ be projections. Put $A=A_\Phi(P,Q)$ and $P\Rightarrow_\Phi Q=E^A([0,\infty))$. We call this value the \emph{truth projection} produced by $\Phi$ for the pair $(P,Q)$.
\end{definition}

The measurable matrix field form is intrinsic to each algebra $W^*(P,Q)$. Definition~\ref{def:borel-symbol} imposes a uniformity requirement because the same symbol acts on every pair. When $\Phi$ is continuous, the Raeburn--Sinclair model \cite{raeburn-sinclair} gives a unique self-adjoint element $a_\Phi\in C^*(p,q)$. Let $\pi_{P,Q}$ denote the unique unital $*$-representation determined by $\pi_{P,Q}(p)=P$ and $\pi_{P,Q}(q)=Q$. Then $A_\Phi(P,Q)=\pi_{P,Q}(a_\Phi)$ for every pair $P,Q$. On the generic summand, both sides apply the same matrix field to the spectral decomposition of the central parameter. On the four commuting sectors, both constructions use the corresponding diagonal endpoint entries. Thus continuous symbols and self-adjoint elements of the universal algebra describe the same uniform constructions.

The spectral sign implication belongs to this class. The map $x\mapsto b_x$ is continuous, and the endpoint values $b_0=\operatorname{diag}(-1,1)$ and $b_1=0$ are diagonal. Hence $\Phi_{\mathrm{sp}}(x)=b_x$ is a continuous symbol, with $A_{\Phi_{\mathrm{sp}}}(P,Q)=Q-P$. Its symbol implication is exactly $\Rightarrow_{\mathrm{sp}}$.

Three related descriptions will be used below. In the first, the symbol $\Phi$ specifies the uniform self-adjoint construction to which the spectral threshold is applied. The second uses a Borel set $Z$, introduced in Subsection~\ref{subsec:selectors}, to specify the resulting truth projection as a spectral selector of $Q-P$. In the continuous formulation using the universal algebra, a self-adjoint element $a\in C^*(p,q)$ specifies the same type of construction through its representations. These descriptions appear successively as the argument becomes more rigid.

\subsection{Branch Reduction}
\label{subsec:branch}

The first step reduces every admissible truth projection on a generic fiber to one of the two spectral branches of $Q-P$. The measurable branch choices are classified in Subsection~\ref{subsec:selectors}, and continuity is imposed in Subsection~\ref{subsec:continuous}. Only (E), (CN), and (F) are used in this part of the argument. The stronger equalities for modus ponens and modus tollens, together with condition (NG), are properties of the spectral sign implication and are not assumed below. The proof begins with the endpoint constraints and the symmetry of a generic fiber.

The endpoint values of a symbol are already strongly constrained by (E) and (F).

\begin{lemma}[Endpoint orientation]
\label{lem:endpoint-orientation}
Let $\Phi$ be a bounded Borel symbol for two projections and suppose $\Rightarrow_\Phi$ satisfies (E) and (F). Write $\Phi(0)=\operatorname{diag}(c_{10},c_{01})$ and $\Phi(1)=\operatorname{diag}(c_{11},c_{00})$. Then $c_{10}<0$ and $c_{01},c_{11},c_{00}\ge0$.
\end{lemma}

\begin{proof}
The pair $(I,0)$ corresponds to the first slot at $x=0$. Condition (F) requires its truth projection to be zero, so $c_{10}<0$. The pair $(0,I)$ corresponds to the second slot at $x=0$ and satisfies $0\le I$. Condition (E) therefore requires $c_{01}\ge0$. The pairs $(I,I)$ and $(0,0)$ correspond to the two slots at $x=1$ and satisfy the order relation, so (E) requires $c_{11},c_{00}\ge0$.
\end{proof}

The generic fibers also have a symmetry under complementation.

\begin{lemma}[Complement swap]
\label{lem:complement-swap}
Let $x\in(0,1)$, put $t=\sqrt{1-x}$, and let $v_x=b_x/t$. Then $v_x$ is a self-adjoint unitary and $v_xp_xv_x=I-q_x$, $v_xq_xv_x=I-p_x$.
\end{lemma}

\begin{proof}
Since $b_x^2=t^2I$, the operator $v_x$ is a self-adjoint unitary. Multiplying the matrices in Definition~\ref{def:canonical-fibers}, we verify the two conjugation formulas.
\end{proof}

This symmetry is the key step in reducing the possible truth projections.

\begin{theorem}[Branch reduction]
\label{thm:branch-reduction}
Let $\Phi$ be a bounded Borel symbol for two projections. Suppose $\Rightarrow_\Phi$ satisfies (E), (CN), and (F). For every $x\in(0,1)$, put $t=\sqrt{1-x}$, $R_x^+=E^{b_x}(\{t\})$, and $R_x^-=E^{b_x}(\{-t\})$. Then the truth projection $E^{\Phi(x)}([0,\infty))$ is either $R_x^+$ or $R_x^-$. The choice defines a Borel map $\varepsilon_\Phi:(0,1)\to\{+,-\}$, called the \emph{branch function} of $\Phi$.
\end{theorem}

\begin{proof}
Let $x\in(0,1)$. For the canonical pair $(p_x,q_x)$, the construction of $A_\Phi$ uses multiplicity space $\mathbb C$ and central parameter $H=x$, so $A_\Phi(p_x,q_x)=\Phi(x)$. Put $R_x=E^{\Phi(x)}([0,\infty))$. Since $0<x<1$, the pair $(p_x,q_x)$ is a generic fiber. In particular, $p_x\nleq q_x$, so (E) requires $R_x\ne I$. The pair is not $(I,0)$, so (F) requires $R_x\ne0$. Hence $R_x$ has rank one.

By Lemmas~\ref{lem:complement-swap} and \ref{lem:symbol-covariance}, $(I-q_x)\Rightarrow_\Phi(I-p_x)=v_xR_xv_x$. Condition (CN) requires $R_x=v_xR_xv_x$. Thus $R_x$ commutes with $v_x$. The two eigenspaces of $v_x$ are exactly the two eigenspaces of $b_x$. A nontrivial projection commuting with $v_x$ must therefore equal $R_x^+$ or $R_x^-$.

Define $\varepsilon_\Phi(x)=+$ in the first case and $\varepsilon_\Phi(x)=-$ in the second. The map $x\mapsto\Phi(x)$ is Borel. To verify the required measurability of the spectral threshold, let $g_n:\mathbb R\to[0,1]$ be continuous, equal to $0$ on $(-\infty,-1/n]$, equal to $1$ on $[0,\infty)$, and linear on $[-1/n,0]$. For every self-adjoint $2\times2$ matrix $C$, the matrices $g_n(C)$ converge to $E^C([0,\infty))$. The map $C\mapsto g_n(C)$ is continuous for each fixed $n$. Hence $C\mapsto E^C([0,\infty))$ is Borel, and therefore $x\mapsto R_x$ is Borel. The projection $R_x^+$ is continuous on $(0,1)$. Hence the function $x\mapsto\operatorname{Tr}(R_xR_x^+)$ is Borel and takes only the values $0$ and $1$. Its value is $1$ exactly on the positive branch. Thus $\varepsilon_\Phi$ is Borel.
\end{proof}

The branch reduction also constrains the defining observable on generic fibers.

\begin{proposition}[Generic observable form]
\label{prop:generic-observable-form}
Under the hypotheses of Theorem~\ref{thm:branch-reduction}, one has $\Phi(x)b_x=b_x\Phi(x)$ for every $x\in(0,1)$. Consequently, there are Borel scalar functions $\alpha,\beta:(0,1)\to\mathbb R$ such that $\Phi(x)=\alpha(x)I+\beta(x)b_x$.
\end{proposition}

\begin{proof}
The projection $R_x=E^{\Phi(x)}([0,\infty))$ has rank one. Both $R_x\mathbb C^2$ and its orthogonal complement are spectral subspaces of the self-adjoint matrix $\Phi(x)$. Hence $\Phi(x)$ is diagonal with respect to the decomposition $R_x\mathbb C^2\oplus(I-R_x)\mathbb C^2$. Theorem~\ref{thm:branch-reduction} says that this is also the eigenspace decomposition of $b_x$. Thus the two matrices commute. Since $b_x$ is not scalar, its commutant in $M_2(\mathbb C)$ is $\operatorname{span}\{I,b_x\}$. Writing $\Phi(x)=\alpha(x)I+\beta(x)b_x$ and using $\operatorname{Tr}b_x=0$ and $b_x^2=(1-x)I$, we obtain $\alpha(x)=\frac12\operatorname{Tr}\Phi(x)$ and $\beta(x)=\operatorname{Tr}(b_x\Phi(x))/(2(1-x))$. These formulas show that $\alpha$ and $\beta$ are Borel on $(0,1)$.
\end{proof}

At the observable level, the endpoint corresponding to the spectral value zero of $Q-P$ carries an additional degree of freedom. The two sectors $P\wedge Q$ and $P^\perp\wedge Q^\perp$ both have spectral value zero for $Q-P$, while a general matrix symbol may assign them different nonnegative scalar values. Both scalar values contribute to the truth projection because they are nonnegative. Thus, after the generic reduction above, a possible difference between these two zero sectors is the only remaining obstruction to representing the observable as a Borel function of $Q-P$. It does not affect the truth projection.

\section{Classification}
\label{sec:classification}

The branch reduction restricts each generic fiber to one of the two spectral branches of $Q-P$. The next stage treats measurable branch selections and then determines the additional restriction imposed by continuity.

\subsection{Borel Classification}
\label{subsec:selectors}

A Borel subset of $[-1,1]$ can represent a Borel branch function. We first define the sets that encode complete branch choices.

\begin{definition}[Borel transversal]
\label{def:borel-transversal}
A Borel set $Z\subseteq[-1,1]$ is a \emph{Borel transversal} if $0,1\in Z$, $-1\notin Z$, and exactly one of $t$ and $-t$ belongs to $Z$ for every $0<t<1$.
\end{definition}

The selector construction is defined for arbitrary Borel sets before conditions (E) and (F) are imposed.

\begin{definition}[Borel spectral selector]
\label{def:borel-selector}
Let $Z\subseteq[-1,1]$ be Borel. The operation $P\Rightarrow_ZQ=E^{Q-P}(Z)$ is called the \emph{Borel spectral selector} associated with $Z$.
\end{definition}

Every Borel spectral selector satisfies (CN) because the comparison operator is unchanged by $(P,Q)\mapsto(Q^\perp,P^\perp)$. The next two propositions show that the endpoint requirements and the branch condition in Definition~\ref{def:borel-transversal} are exactly those imposed jointly by (E) and (F). We state the two criteria separately.

\begin{proposition}[Entailment criterion]
\label{prop:selector-e}
The operation $\Rightarrow_Z$ satisfies (E) if and only if $0,1\in Z$, $-1\notin Z$, and $|\{t,-t\}\cap Z|\le1$ for every $0<t<1$.
\end{proposition}

\begin{proof}
Assume first that (E) holds. The pairs $(0,0)$, $(0,I)$, and $(I,0)$ show respectively that $0\in Z$, $1\in Z$, and $-1\notin Z$. Let $0<t<1$. There is a generic qubit pair whose comparison operator has eigenvalues $\pm t$. If both values belonged to $Z$, the truth projection would be $I$, contrary to (E). Thus at most one of $t$ and $-t$ belongs to $Z$.

Conversely, assume the stated conditions and suppose $E^{Q-P}(Z)=I$. Decompose the pair into the four commuting sectors and the generic summand. On the generic summand, write $B=Q-P$. From the equality, $E^B(Z^c)=0$. By Lemma~\ref{lem:spectral-symmetry}, $E^B(-Z^c)=0$. The condition that at most one of $t$ and $-t$ belongs to $Z$ implies $Z^c\cup(-Z^c)\supseteq(-1,1)\setminus\{0\}$. By Lemma~\ref{lem:exceptional-sectors}, the singleton spectral projections at $0$, $1$, and $-1$ vanish on the generic summand. Therefore the generic summand must be zero.

Only the commuting sectors remain. The value of $B$ on $\mathcal H_{10}$ is $-1$, which does not belong to $Z$. Since the truth projection is $I$, this sector must be zero. Thus $P\le Q$. The reverse implication follows because $P\le Q$ makes $Q-P$ a projection with spectrum contained in $\{0,1\}$, and both values belong to $Z$.
\end{proof}

The corresponding criterion for the zero truth value is complementary.

\begin{proposition}[Falsity criterion]
\label{prop:selector-f}
The operation $\Rightarrow_Z$ satisfies (F) if and only if $0,1\in Z$, $-1\notin Z$, and $|\{t,-t\}\cap Z|\ge1$ for every $0<t<1$.
\end{proposition}

\begin{proof}
Assume first that (F) holds. The pair $(I,0)$ shows that $-1\notin Z$. If $0\notin Z$, the pair $P=Q=0$ would have truth projection zero. If $1\notin Z$, the pair $P=0$, $Q=I$ would have truth projection zero. Hence $0,1\in Z$. Let $0<t<1$ and use a generic qubit pair with comparison eigenvalues $\pm t$. If neither value belonged to $Z$, the truth projection would be zero, contrary to (F). Thus at least one of $t$ and $-t$ belongs to $Z$.

Conversely, assume the stated conditions and suppose $E^{Q-P}(Z)=0$. Decompose the pair into the four commuting sectors and the generic summand. On the generic summand, write $B=Q-P$. From the equality, $E^B(Z)=0$. By Lemma~\ref{lem:spectral-symmetry}, $E^B(-Z)=0$. The condition that at least one of $t$ and $-t$ belongs to $Z$ implies $Z\cup(-Z)\supseteq(-1,1)\setminus\{0\}$. Again by Lemma~\ref{lem:exceptional-sectors}, the singleton spectral projections at $0$, $1$, and $-1$ vanish on the generic summand. Hence the generic summand is zero.

On the commuting sectors, $0\in Z$ eliminates $\mathcal H_{11}$ and $\mathcal H_{00}$, while $1\in Z$ eliminates $\mathcal H_{01}$. The only possible sector is $\mathcal H_{10}$. Thus $P=I$ and $Q=0$. The converse follows from $-1\notin Z$.
\end{proof}

From the preceding propositions, we obtain the exact Borel classification.

\begin{theorem}[Borel transversal classification]
\label{thm:borel-classification}
The operation $\Rightarrow_Z$ satisfies (E), (CN), and (F) if and only if $Z$ is a Borel transversal.
\end{theorem}

\begin{proof}
Condition (CN) is automatic for every $Z$. The remaining equivalence is the conjunction of Propositions~\ref{prop:selector-e} and \ref{prop:selector-f}.
\end{proof}

Every Borel transversal is realized by a symbol in the class introduced in Definition~\ref{def:borel-symbol}.

\begin{proposition}[Borel realization]
\label{prop:borel-realization}
Let $Z$ be a Borel transversal. Then there is a bounded Borel symbol $\Phi_Z$ for two projections such that $P\Rightarrow_{\Phi_Z}Q=P\Rightarrow_ZQ$ for every pair $P,Q$.
\end{proposition}

\begin{proof}
Define $f_Z:[-1,1]\to\{-1,1\}$ by $f_Z(t)=1$ when $t\in Z$ and $f_Z(t)=-1$ otherwise. Put $\Phi_Z(x)=f_Z(b_x)$ for $x\in[0,1]$. Let $0\le x<1$ and put $t=\sqrt{1-x}$. Define $R_x^\pm=\frac12(I\pm b_x/t)$ on this interval. Then $\Phi_Z(x)=f_Z(t)R_x^++f_Z(-t)R_x^-$. The coefficients are Borel functions of $x$, and $R_x^\pm$ are continuous on $[0,1)$. Hence the entries of $\Phi_Z$ are Borel on $[0,1)$. At $x=1$, $\Phi_Z(1)=f_Z(0)I$, so $\Phi_Z$ is a bounded Borel map on $[0,1]$. Moreover, $b_0=\operatorname{diag}(-1,1)$ and $b_1=0$. Thus both endpoint values of $\Phi_Z$ are diagonal, and $\Phi_Z$ is a bounded Borel symbol.

On the generic summand, fiberwise functional calculus yields $A_{\Phi_Z}(P,Q)=f_Z(Q-P)$. The same equality holds on the four commuting sectors because the endpoint values of $Q-P$ are $-1$, $1$, and $0$, exactly as prescribed by the endpoint slots of the symbol. Since $f_Z$ is nonnegative exactly on $Z$, the spectral mapping theorem yields $E^{f_Z(Q-P)}([0,\infty))=E^{Q-P}(Z)$. Hence $P\Rightarrow_{\Phi_Z}Q=P\Rightarrow_ZQ$ for every pair.
\end{proof}

The fixed threshold used in the construction can also be justified within the Borel class.

\begin{remark}[Choice of spectral threshold]
\label{rem:borel-threshold}
The threshold $[0,\infty)$ causes no loss of generality for the bounded Borel symbol class. Let $W\subseteq\mathbb R$ be Borel and suppose the truth projection is defined by $E^{A_\Phi(P,Q)}(W)$. Define $h_W$ to be $1$ on $W$ and $-1$ on its complement, and put $\Psi(x)=h_W(\Phi(x))$ by matrix functional calculus. The map $C\mapsto h_W(C)$ on $M_2(\mathbb C)_{\mathrm{sa}}$ is Borel. Indeed, on the open set of nonscalar matrices the ordered eigenvalues and their spectral projections are continuous, while $h_W(\lambda I)=h_W(\lambda)I$ on scalar matrices. Since functional calculus preserves diagonal matrices, $\Psi(0)$ and $\Psi(1)$ are diagonal. Hence $\Psi$ is again a bounded Borel symbol. The symbol calculus commutes with Borel functional calculus on the generic matrix field and on each endpoint sector, so $A_\Psi(P,Q)=h_W(A_\Phi(P,Q))$. Therefore $E^{A_\Phi(P,Q)}(W)=E^{A_\Psi(P,Q)}([0,\infty))$. This reduction need not preserve continuity, so it does not reduce an arbitrary Borel threshold to the same continuous symbol class.
\end{remark}

Varying the Borel transversal, we obtain many distinct operations.

\begin{corollary}[Cardinality of Borel selectors]
\label{cor:borel-cardinality}
There are continuum many distinct Borel spectral selectors satisfying (E), (CN), and (F).
\end{corollary}

\begin{proof}
Let $T\subseteq(0,1)$ be Borel and put $Z_T=\{0,1\}\cup T\cup\{-t:t\in(0,1)\setminus T\}$. Each $Z_T$ is a Borel transversal, so Theorem~\ref{thm:borel-classification} gives a Borel spectral selector satisfying (E), (CN), and (F). There are continuum many Borel subsets of $(0,1)$. If $T_1\ne T_2$, choose $t$ in their symmetric difference and take a generic qubit pair with comparison eigenvalues $\pm t$. The corresponding truth projections are the projections onto opposite eigenspaces, so the operations are distinct.
\end{proof}

Combining the branch reduction with the selector criterion, we obtain the complete classification for the original bounded Borel symbol class.

\begin{theorem}[Borel symbol classification]
\label{thm:symbol-classification}
An operation arising from a bounded Borel symbol for two projections satisfies (E), (CN), and (F) if and only if it is a Borel spectral selector $\Rightarrow_Z$ for some Borel transversal $Z$.
\end{theorem}

\begin{proof}
Suppose first that $\Phi$ is a bounded Borel symbol and $\Rightarrow_\Phi$ satisfies (E), (CN), and (F). Define $Z$ on $(-1,1)\setminus\{0\}$ by requiring $t\in Z$ when $\varepsilon_\Phi(1-t^2)=+$ and $-t\in Z$ when $\varepsilon_\Phi(1-t^2)=-$ for each $t\in(0,1)$. The branch function is Borel by Theorem~\ref{thm:branch-reduction}, so the resulting subset of $(-1,1)\setminus\{0\}$ is Borel. Add $0$ and $1$, and omit $-1$. The endpoint truth values follow from Lemma~\ref{lem:endpoint-orientation}. On the generic summand, the direct integral form yields $E^{A_\Phi(P,Q)}([0,\infty))=E^{Q-P}(Z)$. Combining this equality with the endpoint truth values, we obtain $P\Rightarrow_\Phi Q=P\Rightarrow_ZQ$ on every summand. The set $Z$ is a Borel transversal by construction.

Conversely, let $Z$ be a Borel transversal. Proposition~\ref{prop:borel-realization} provides a bounded Borel symbol $\Phi_Z$ whose operation equals $\Rightarrow_Z$. Theorem~\ref{thm:borel-classification} shows that this operation satisfies (E), (CN), and (F).
\end{proof}

\subsection{Continuous Rigidity}
\label{subsec:continuous}

Continuity removes the remaining freedom in the branch choice.

\begin{theorem}[Continuous uniqueness]
\label{thm:continuous-uniqueness}
Let $\Phi:[0,1]\to M_2(\mathbb C)_{\mathrm{sa}}$ be a continuous symbol for two projections. If $\Rightarrow_\Phi$ satisfies (E), (CN), and (F), then $P\Rightarrow_\Phi Q=P\Rightarrow_{\mathrm{sp}}Q$ for all projections $P,Q$.
\end{theorem}

\begin{proof}
Let $x\in(0,1)$. By Theorem~\ref{thm:branch-reduction} and Proposition~\ref{prop:generic-observable-form}, $\Phi(x)b_x=b_x\Phi(x)$. Write $\Phi(x)=\lambda_+(x)R_x^++\lambda_-(x)R_x^-$. The functions $\lambda_\pm(x)=\operatorname{Tr}(\Phi(x)R_x^\pm)$ are continuous on $(0,1)$.

Conditions (E) and (F) say that exactly one of $\lambda_+(x)$ and $\lambda_-(x)$ is nonnegative at each $x$. Put $A_+=\{x\in(0,1):\lambda_+(x)\ge0\}$. By continuity, $A_+$ is closed. Its complement equals $\{x\in(0,1):\lambda_-(x)\ge0\}$ and is also closed. Thus $A_+$ is open and closed in the connected interval $(0,1)$.

When $0<x<1$, $R_x^\pm=\frac12(I\pm b_x/\sqrt{1-x})$. Hence $R_x^+$ and $R_x^-$ extend continuously to $x=0$. The positive eigenspace of $b_0$ corresponds to the second diagonal slot, and the negative eigenspace corresponds to the first. Therefore the formulas for $\lambda_+$ and $\lambda_-$ extend continuously to $x=0$. By Lemma~\ref{lem:endpoint-orientation}, $\lambda_+(0)\ge0$ and $\lambda_-(0)<0$. The strict inequality and continuity give $\lambda_-(x)<0$ for all sufficiently small positive $x$. Since exactly one eigenvalue is nonnegative, such $x$ belong to $A_+$. Hence $A_+$ is nonempty and must equal $(0,1)$.

The generic truth projection is therefore the positive spectral projection of $Q-P$ on every generic fiber. On the two zero difference sectors, (E) makes the truth projection equal to the identity. The Boolean endpoint sectors have the same orientation by Lemma~\ref{lem:endpoint-orientation}. Hence the global truth projection is $E^{Q-P}([0,\infty))$.
\end{proof}

The following example shows that a continuous branch switch is incompatible with the sign conditions required by (E) and (F).

\begin{example}[Continuous switching]
\label{ex:continuous-switching}
Consider $\psi(t)=t(t^2-1/4)(t^2-3/4)$. When $t>0$, its sign is positive on $(0,1/2)$, negative on $(1/2,\sqrt3/2)$, and positive on $(\sqrt3/2,1]$. The corresponding branch choice therefore changes from the positive branch to the negative branch and then back to the positive branch, while $\psi(1)>0$. At $t=1/2$ and $t=\sqrt3/2$, both $\psi(t)$ and $\psi(-t)$ vanish. Since $t=\sin\theta$ for the canonical line pair, these values correspond to angles $\pi/6$ and $\pi/3$, respectively. At either value of $t$, a generic qubit pair satisfies $E^{\psi(Q-P)}([0,\infty))=I$, so (E) fails.

A Borel selector can realize the same branch pattern by jumps. Let $T=(0,1/2)\cup[\sqrt3/2,1)$ and form $Z_T$ as in Corollary~\ref{cor:borel-cardinality}. Exactly one of $t$ and $-t$ belongs to $Z_T$ for every $0<t<1$. The resulting Borel operation satisfies (E) and (F) by Theorem~\ref{thm:borel-classification}.
\end{example}

By contrast, a Borel branch change requires no intermediate point at which (E) or (F) fails.

\subsection{Observable Reduction}
\label{subsec:observable}

Theorem~\ref{thm:continuous-uniqueness} characterizes the truth projection. Within the universal $C^*$-algebra, the same logical hypotheses imply a stronger statement about the self-adjoint element to which the spectral threshold is applied.

Recall from Subsection~\ref{subsec:symbols} that the Raeburn--Sinclair model represents $C^*(p,q)$ by continuous $M_2(\mathbb C)$-valued functions with diagonal endpoint values \cite{raeburn-sinclair}. We use this representation below.

\begin{theorem}[Continuous observable reduction]
\label{thm:observable-reduction}
Let $a=a^*\in C^*(p,q)$ and define $P\Rightarrow_aQ=E^{\pi_{P,Q}(a)}([0,\infty))$ for each pair of projections, where $\pi_{P,Q}$ is the universal representation with $\pi_{P,Q}(p)=P$ and $\pi_{P,Q}(q)=Q$. If $\Rightarrow_a$ satisfies (E), (CN), and (F), then there is a continuous real function $\psi:[-1,1]\to\mathbb R$ such that $a=\psi(q-p)$.
\end{theorem}

\begin{proof}
Write $a_x$ for the value of $a$ at $x$ in the Raeburn--Sinclair model. The matrix field $x\mapsto a_x$ is a continuous symbol for two projections. By the correspondence described in Subsection~\ref{subsec:symbols}, its symbol implication is $\Rightarrow_a$, which satisfies (E), (CN), and (F) by hypothesis. Proposition~\ref{prop:generic-observable-form} therefore gives $a_xb_x=b_xa_x$ on every generic fiber.

Put $t=\sqrt{1-x}$ for $0<x<1$. Let $R_x^\pm$ be the spectral projections of $b_x$ for $\pm t$. Define $\psi(t)=\operatorname{Tr}(a_xR_x^+)$ and $\psi(-t)=\operatorname{Tr}(a_xR_x^-)$. Since $a_x$ commutes with $b_x$, we can write $a_x=\psi(t)R_x^++\psi(-t)R_x^-$. These formulas show continuity on $(0,1)$ and $(-1,0)$. At $x=0$, the projections $R_x^\pm$ converge to the two diagonal coordinate projections. Since $a_0$ is diagonal, the same formulas extend $\psi$ continuously to $1$ and $-1$.

It remains to treat $t=0$, which corresponds to $x=1$. Calculating directly, we obtain $b_x/t\to\begin{pmatrix}0&1\\1&0\end{pmatrix}$. Thus $R_x^\pm$ converge to $R^\pm=\frac12\begin{pmatrix}1&\pm1\\\pm1&1\end{pmatrix}$. Since $a_x$ and $R_x^\pm$ converge, the limits $c_+=\lim_{t\downarrow0}\psi(t)$ and $c_-=\lim_{t\downarrow0}\psi(-t)$ exist. The endpoint value is $a_1=c_+R^++c_-R^-$. The universal model requires $a_1$ to be diagonal. Its two off-diagonal entries are both $(c_+-c_-)/2$, so $c_+=c_-$. Define $\psi(0)$ to be this common value. Then $\psi$ is continuous on $[-1,1]$ and $a_x=\psi(b_x)$ on every canonical fiber, including both endpoints. Hence $a=\psi(q-p)$.
\end{proof}

Conditions (E) and (F) permit continuous self-adjoint observables for two projections with components outside the commutant of $q-p$. Subsection~\ref{subsec:independence} presents an explicit family. Theorem~\ref{thm:observable-reduction} shows that (CN) excludes these components.

\subsection{Axiom Independence}
\label{subsec:independence}

The three conditions in the characterization have distinct roles. We now give witnesses showing that none follows from the other two within the class of continuous constructions arising from the universal algebra.

\begin{proposition}[Independence witnesses]
\label{prop:independence}
Conditions (E), (CN), and (F) are mutually independent for continuous constructions arising from the universal algebra generated by two projections.
\end{proposition}

\begin{proof}
To omit (E), define $\psi_E(t)=t+1-t^2$ and use $a_E=\psi_E(q-p)$. Put $Z_E=\{t\in[-1,1]:\psi_E(t)\ge0\}$. By functional calculus, $E^{\psi_E(Q-P)}([0,\infty))=E^{Q-P}(Z_E)$. At the endpoints, $\psi_E(0)=1$, $\psi_E(1)=1$, and $\psi_E(-1)=-1$. When $0<t<1$, $\psi_E(t)>0$, so at least one of $t$ and $-t$ belongs to $Z_E$. Proposition~\ref{prop:selector-f} establishes (F). Condition (CN) holds because the comparison operator $Q-P$ is unchanged by contraposition. For all sufficiently small positive $t$, $\psi_E(-t)=1-t-t^2\ge0$. Both $t$ and $-t$ then belong to $Z_E$, so Proposition~\ref{prop:selector-e} shows that (E) fails.

To omit (CN), put $b=q-p$, $s=p+q-I$, and $a_\lambda=b+\lambda ibs$ for a fixed real $\lambda\ne0$. Since $bs=-sb$, the element $ibs$ is self-adjoint, so $a_\lambda$ is self-adjoint. For a pair of projections $P,Q$, write $B=\pi_{P,Q}(b)=Q-P$, $S=\pi_{P,Q}(s)=P+Q-I$, and $A_\lambda=\pi_{P,Q}(a_\lambda)=B+\lambda iBS$. On a generic summand, let $U=\operatorname{sgn}(S)$. By Lemma~\ref{lem:spectral-symmetry}, $UBU=-B$. Since $U$ is a function of $S$, $USU=S$ as well. Hence $UA_\lambda U=-A_\lambda$. Using the anticommutation of $B$ and $S$, we obtain $A_\lambda^2=B^2(I+\lambda^2S^2)$. On the generic summand, $\ker B=0$, while $I+\lambda^2S^2$ is invertible. Thus $\ker A_\lambda=0$.

Let $R=E^{A_\lambda}([0,\infty))$ on a nonzero generic summand. Unitary covariance and $UA_\lambda U=-A_\lambda$ give $URU=E^{A_\lambda}(({-\infty},0])$. Since $\ker A_\lambda=0$, the latter projection is $I-R$. Therefore $R$ is neither $0$ nor $I$. On the four commuting sectors, $A_\lambda$ is $0$ on $\mathcal H_{11}$ and $\mathcal H_{00}$, $-I$ on $\mathcal H_{10}$, and $I$ on $\mathcal H_{01}$. If the global truth projection is $I$, the generic summand must vanish and $\mathcal H_{10}=0$, which is equivalent to $P\le Q$. Conversely, $P\le Q$ leaves only sectors on which the truth value is $I$. Thus (E) holds. If the global truth projection is $0$, the generic summand and the sectors $\mathcal H_{11},\mathcal H_{01},\mathcal H_{00}$ must vanish. Hence $P=I$ and $Q=0$. The converse is immediate, so (F) holds.

Under the complement swap $(P,Q)\mapsto(Q^\perp,P^\perp)$, the operator $B$ is fixed and $S$ changes sign. Thus the represented operator $A_\lambda$ changes to $A_{-\lambda}=B-\lambda iBS$. On a generic fiber, $B^2=(1-x)I$ and $S^2=xI$ with $0<x<1$, so $BS\ne0$ and $A_\lambda\ne A_{-\lambda}$. Both matrices are nonzero traceless self-adjoint matrices and have the same square. If $A$ is a nonzero traceless self-adjoint $2\times2$ matrix and $R_A$ is its positive spectral projection, then $A=\|A\|(2R_A-I)$. Hence two such matrices with the same norm and the same positive spectral projection must be equal. The positive spectral projections of $A_\lambda$ and $A_{-\lambda}$ are therefore different, so (CN) fails.

To omit (F), define $\psi_F(t)=t-2|t|\sin(\pi(1-t^2))$ and use $a_F=\psi_F(q-p)$. Put $Z_F=\{t\in[-1,1]:\psi_F(t)\ge0\}$. By functional calculus, $E^{\psi_F(Q-P)}([0,\infty))=E^{Q-P}(Z_F)$ for every pair of projections. At the endpoints, $\psi_F(0)=0$, $\psi_F(1)=1$, and $\psi_F(-1)=-1$. When $0<t<1$, $\psi_F(-t)=-t(1+2\sin(\pi(1-t^2)))<0$. Hence at most one of $t$ and $-t$ belongs to $Z_F$, so Proposition~\ref{prop:selector-e} establishes (E). Condition (CN) holds because the comparison operator $Q-P$ is unchanged by contraposition. Whenever $\sin(\pi(1-t^2))>1/2$, $\psi_F(t)<0$ as well. Then neither $t$ nor $-t$ belongs to $Z_F$, and Proposition~\ref{prop:selector-f} shows that (F) fails.
\end{proof}

The witness for (CN) also clarifies its role in the characterization. On a generic fiber, the $iBS$ term in $A_\lambda$ makes the positive spectral projection differ from both spectral branch projections of $Q-P$. Theorem~\ref{thm:observable-reduction} shows that contraposition excludes this additional component once (E) and (F) are retained.

\section{Variational Characterization}
\label{sec:variational}

The discussion now shifts from logical conditions and regularity to an independent extremal characterization and its Helstrom interpretation. Let $\mathcal H$ be finite dimensional throughout this section.

\subsection{Projection Optimization}
\label{subsec:projection-optimization}

The starting point is the optimization problem for a general self-adjoint operator.

\begin{lemma}[Projection maximizers]
\label{lem:projection-maximizers}
Let $A=A^*$, and let $A_+$ and $A_-$ denote its positive and negative parts. Let $R$ range over the projections on $\mathcal H$. The maximum of $\operatorname{Tr}(RA)$ is $\operatorname{Tr}(A_+)$. The maximizing projections are exactly those satisfying $E^A((0,\infty))\le R\le E^A([0,\infty))$.
\end{lemma}

\begin{proof}
Write $A=A_+-A_-$ with $A_+,A_-\ge0$ and orthogonal supports. Then $\operatorname{Tr}(RA)=\operatorname{Tr}(RA_+)-\operatorname{Tr}(RA_-)\le\operatorname{Tr}(A_+)$. Equality holds exactly when $\operatorname{Tr}((I-R)A_+)=0$ and $\operatorname{Tr}(RA_-)=0$. Since the operators $A_+^{1/2}(I-R)A_+^{1/2}$ and $A_-^{1/2}RA_-^{1/2}$ are positive, trace zero is equivalent to their vanishing. Thus $R$ acts as the identity on the positive spectral subspace and as zero on the negative spectral subspace. Its action on the kernel is arbitrary. This is equivalent to the stated interval of projections.
\end{proof}

We now apply the lemma to the projection difference.

\begin{theorem}[Variational implication]
\label{thm:variational-implication}
When $\mathcal H$ is finite dimensional and $P,Q$ are projections, the spectral sign implication is the largest projection among all projections $R$ that maximize $\operatorname{Tr}(R(Q-P))$.
\end{theorem}

\begin{proof}
Apply Lemma~\ref{lem:projection-maximizers} to $A=Q-P$. The maximal member of the interval of maximizing projections is $E^{Q-P}([0,\infty))$.
\end{proof}

The dual conjunction selects the opposite extremal optimizer in the corresponding comparison problem.

\begin{corollary}[Variational dual]
\label{cor:variational-dual}
When $\mathcal H$ is finite dimensional and $P,Q$ are projections, the projection $P*_{\mathrm{sp}}Q$ is the smallest projection maximizing $\operatorname{Tr}(R(P-Q^\perp))$.
\end{corollary}

\begin{proof}
By Proposition~\ref{prop:dual-formula}, $P*_{\mathrm{sp}}Q=E^{P-Q^\perp}((0,\infty))$. Lemma~\ref{lem:projection-maximizers} shows that this is the smallest maximizing projection.
\end{proof}

\subsection{Helstrom Interpretation}
\label{subsec:helstrom}

The variational result has a standard binary discrimination interpretation \cite{helstrom,holevo}. Let nonzero projections $P,Q$ have ranks $m,n$. Define the normalized subspace states $\rho_P=P/m$ and $\rho_Q=Q/n$, and choose priors $\pi_P=m/(m+n)$ and $\pi_Q=n/(m+n)$.

\begin{corollary}[Helstrom interpretation]
\label{cor:helstrom}
Under the preceding choice of subspace states and priors, $P\Rightarrow_{\mathrm{sp}}Q$ is the largest among the acceptance projections of the projective tests that are optimal in the corresponding Helstrom problem.
\end{corollary}

\begin{proof}
The weighted difference is $\pi_Q\rho_Q-\pi_P\rho_P=(Q-P)/(m+n)$. If the projective test $R$ accepts $\rho_Q$, its success probability is $\pi_Q\operatorname{Tr}(R\rho_Q)+\pi_P\operatorname{Tr}((I-R)\rho_P)=\pi_P+\operatorname{Tr}(R(\pi_Q\rho_Q-\pi_P\rho_P))$. Maximizing this probability is therefore equivalent to maximizing $\operatorname{Tr}(R(Q-P))$. The largest optimizer is the projection described in Theorem~\ref{thm:variational-implication}.
\end{proof}

Consider a generic pair of line projections $P_0,P_\alpha$. Since both projections have rank one, the priors chosen above are equal. The eigenvalues of $Q-P$ are $\pm\sin\alpha$. The optimal success probability is therefore $(1+\sin\alpha)/2$. By Proposition~\ref{prop:signed-angle}, the acceptance projection is $P_{\pi/4+\alpha/2}$.

\section{Quantum Set Connection}
\label{sec:qst}

The final application concerns the threshold semantics of quantum set theory. We first place the spectral sign implication within Ozawa's framework and define the corresponding threshold truth value. We then specialize this construction to projection observables and compare its full truth criterion with spectral order for general observables.

\subsection{Threshold Semantics}
\label{subsec:qst-threshold}

Ozawa's framework accommodates nonpolynomial implications through local binary operations \cite{ozawa2021}. Such an operation must take its value on $P,Q$ in the sublogic generated by $P$ and $Q$ and be compatible with restriction to projections commuting with both arguments. A local operation is a quantized implication when it agrees with classical Boolean implication on commuting arguments. Proposition~\ref{prop:locality} establishes the two locality conditions for $\Rightarrow_{\mathrm{sp}}$. In particular, its first assertion places the value in $\operatorname{Proj}(W^*(P,Q))$, which is contained in the sublogic generated by $P$ and $Q$. Proposition~\ref{prop:boolean-reduction} establishes agreement with classical implication on commuting projections. Hence $\Rightarrow_{\mathrm{sp}}$ is a quantized implication. Its De Morgan dual is the conjunction defined in Definition~\ref{def:dual-conjunction}. With self-duality understood as in \cite{ozawa2021}, this choice makes the corresponding interpretation self-dual.

\begin{definition}[Threshold truth value]
\label{def:threshold-truth}
Let $X$ and $Y$ be self-adjoint operators affiliated with $\mathcal M$, and write $E^X(r)=E^X(({-\infty},r])$. Their spectral projections belong to $\operatorname{Proj}(\mathcal M)$. Following the threshold formula of \cite{ozawa2026} for the order truth value of quantum reals, define
\[
T_{\mathrm{sp}}(X,Y)=\bigwedge_{r\in\mathbb Q}\bigl(E^Y(r)\Rightarrow_{\mathrm{sp}}E^X(r)\bigr).
\]
\end{definition}

\subsection{Projection Observables}
\label{subsec:qst-projections}

We first specialize the threshold formula to projection observables.

\begin{proposition}[Projection order truth]
\label{prop:qst-projections}
For all projection observables $P,Q$, one has $T_{\mathrm{sp}}(P,Q)=P\Rightarrow_{\mathrm{sp}}Q$.
\end{proposition}

\begin{proof}
If $P$ is a projection, then $E^P(r)=0$ for $r<0$, $E^P(r)=P^\perp$ for $0\le r<1$, and $E^P(r)=I$ for $r\ge1$. All factors with $r<0$ or $r\ge1$ are $I$, while those with $0\le r<1$ all equal $Q^\perp\Rightarrow_{\mathrm{sp}}P^\perp$. Proposition~\ref{prop:contraposition} shows that this common value equals $P\Rightarrow_{\mathrm{sp}}Q$.
\end{proof}

\subsection{Spectral Order}
\label{subsec:qst-spectral-order}

The full truth criterion can also be verified directly. Olson's spectral order is defined by $X\le_{\mathrm{s}}Y$ exactly when $E^Y(r)\le E^X(r)$ for every real $r$ \cite{olson}.

\begin{proposition}[Spectral order criterion]
\label{prop:qst-spectral-order}
For self-adjoint operators $X,Y$ affiliated with $\mathcal M$, one has $T_{\mathrm{sp}}(X,Y)=I$ if and only if $X\le_{\mathrm{s}}Y$.
\end{proposition}

\begin{proof}
The meet defining $T_{\mathrm{sp}}(X,Y)$ equals $I$ exactly when every factor equals $I$. By condition (E) in Theorem~\ref{thm:basic-laws}, this is equivalent to $E^Y(r)\le E^X(r)$ for every rational $r$.

Suppose these inequalities hold. Let $\lambda\in\mathbb R$ and choose rationals $r_n\downarrow\lambda$. By right continuity of spectral families, $E^Y(\lambda)=\bigwedge_n E^Y(r_n)$ and $E^X(\lambda)=\bigwedge_n E^X(r_n)$. Hence $E^Y(\lambda)\le E^X(\lambda)$. Thus $X\le_{\mathrm{s}}Y$. The converse follows immediately by restricting the defining inequalities of the spectral order to rational values.
\end{proof}

The spectral order criterion agrees with the full truth characterization in Ozawa's treatment of order relations \cite{ozawa2026}. Together, Propositions~\ref{prop:qst-projections} and \ref{prop:qst-spectral-order} place the spectral sign implication within the threshold semantics. The first recovers the implication for projection observables, and the second characterizes full truth for general observables by the standard spectral order.

\section{Conclusion}
\label{sec:conclusion}

The nonnegative spectral sector of $Q-P$ defines an implication directly from operator theory. It agrees with material implication on commuting projections and satisfies the four minimal implicative conditions, contraposition, and the falsity condition. Its qubit value lies outside the ortholattice generated by the input projections, so the construction is nonpolynomial. Hardegree's transitivity condition and consequent monotonicity both fail in general. The failure of consequent monotonicity shows that the operation is not a right residual on $\operatorname{Proj}(M_2(\mathbb C))$.

The characterization separates algebraic and regularity effects. Conditions (E), (CN), and (F) force the truth projection on every generic fiber of a bounded Borel construction for two projections to equal one of the two spectral branch projections of $Q-P$. Under Borel regularity, every Borel transversal is realized. Under continuity, the branch choice is constant, and the endpoint values compatible with classical implication determine the positive branch. Within the universal $C^*$-algebra, the defining observable has the stronger form $\psi(q-p)$.

The variational result provides an independent interpretation. In finite dimensions, the implication is the largest projection maximizing $\operatorname{Tr}(R(Q-P))$. With normalized subspace states and priors proportional to rank, the same projection is the largest among the acceptance projections of the projective tests that are optimal in the Helstrom problem. The De Morgan dual is the smallest optimizer for the complementary comparison. On each principal angle plane of a generic pair, it has the known midpoint geometry for pairs of projections.

The results suggest several extensions. The measurable symbol class motivates the study of regularity assumptions intermediate between Borel measurability and continuity. The same spectral comparison can also be examined for effects or other unsharp propositions. A separate direction concerns general observables and the threshold operation in Section~\ref{sec:qst}, where the full truth value is controlled by spectral order but intermediate truth projections retain dependence on the chosen implication.

\end{document}